\documentclass[reqno,11pt,a4paper]{amsart}
\usepackage{amssymb}
\usepackage{url}
\usepackage{xcolor}
\usepackage[linkbordercolor=blue]{hyperref}

\usepackage{color}
\usepackage{soul}
\allowdisplaybreaks[4]

\usepackage{mathrsfs}
\let\mathcal\mathscr
\usepackage[mathscr]{eucal}

\let\phi=\varphi
\let\kappa=\varkappa

\DeclareMathOperator{\sym}{sym}

\newcommand*{\Ev}{\mathbf{E}}

\newcommand*{\inv}{\mathrm{inv}}

\theoremstyle{theorem}
\newtheorem{proposition}{Proposition}

\theoremstyle{definition}

\theoremstyle{remark}

\newtheorem{case}{Step}

\usepackage[all]{xy}
\SelectTips{cm}{}

\usepackage{mathrsfs}
\let\mathcal\mathscr
\usepackage[mathscr]{eucal}

\newcommand{\cprime}{\/{\mathsurround=0pt$'$}}

\begin{document}
\date{\today}

\title[Veronese web equation]{Nonlocal symmetries, conservation
  laws, and recursion operators of the Veronese web equation}

\author{I.S.~Krasil{\cprime}shchik} \address{Trapeznikov Institute of Control
  Sciences, 65 Profsoyuznaya street, Moscow 117997, Russia \& Independent
  University of Moscow, B. Vlasevsky 11, 119002 Moscow, Russia}
\email{josephkra@gmail.com}\thanks{The work of IK was partially supported by the RFBR Grant 18-29-10013
and IUM-Simons Foundation.}
\author{O.I.~Morozov} \address{Faculty of Applied Mathematics, AGH University
  of Science and Technology, Al. Mickiewicza 30, Cracow 30-059, Poland \&
  Trapeznikov Institute of Control Sciences, 65 Profsoyuznaya street, Moscow
  117997, Russia } \email{morozov{\symbol{64}}agh.edu.pl} \thanks{The work of
  OM was partially supported by the Faculty of Applied Mathematics of AGH UST
  statutory tasks within subsidy of Ministry of Science and Higher Education.}
\author{P.~Voj{\v{c}}{\'{a}}k}
\address{Mathematical Institute, Silesian University in Opava, Na
  Rybn\'{\i}\v{c}ku 1, 746 01 Opava, Czech Republic}
\email{Petr.Vojcak@math.slu.cz}
\thanks{The third author (PV) was supported by
  the Ministry of Education, Youth and Sports of the Czech Republic under the
  project CZ.02.2.69/0.0/0.0/16\_027/0008521.}

\date{\today}

\begin{abstract}
  We study the Veronese web equation
  $u_y u_{tx}+ \lambda u_xu_{ty} - (\lambda+1)u_tu_{xy} =0$ and using its
  isospectral Lax pair construct two infinite series of nonlocal conservation
  laws. In the infinite differential coverings associated to these series, we
  describe the Lie algebras of the corresponding nonlocal symmetries. Finally,
  we construct a recursion operator and explore its action on nonlocal
  shadows. The operator provides a new shadow which serves as a
  master-symmetry.
\end{abstract}

\keywords{Partial differential equations, Veronese web equation, differential
  coverings, Lax pairs, nonlocal symmetries, recursion operators, master
  symmetries}

\subjclass[2010]{35B06}

\maketitle

\tableofcontents
\newpage

\section*{Introduction}

This work finalizes our research of Lax-integrable (i.e., admitting Lax pairs
with non-vanishing spectral parameter) linearly degenerate (in the sense
of~\cite{Fer-Moss-2015}) 3D equations,
see~\cite{B-K-M-V-2014,B-K-M-V-2015,B-K-M-V-2016,B-K-M-V-2018,H-K-M-V-2017,H-K-M-V-2018}. We
deal here with the Veronese web equation (VWE)
\begin{equation}
  \label{eq:1}
  u_y u_{tx} + \lambda\, u_x u_{ty} - (\lambda+1)\,u_t u_{xy} = 0,
\end{equation}
which is a generic case of the so-called ABC-equation, $A+B+C=0$, introduced
in~\cite{Zakh} (see also~\cite{Dun-Kry,Kru-Pan}). Here~$\lambda\neq 0$ is a
real parameter.
This equation determines three-dimensional Veronese webs that appear in the study
of three-dimensional bi-Hamiltonian systems, see \cite{GelfandZakharevich1991} and references therein.
In \cite{Dun-Kry} a one-to-one correspondence between three--dimensional Veronese
webs and Lorentzian Einstein--Weyl structures of hyper-CR type was found.
The latter  are parametrized by the solutions of the hyper-CR equation
\begin{equation}
u_{yy}=u_{tx}+u_y\,u_{xx}-u_x\,u_{xy},
\label{Pavlov_eq}
\end{equation}
which is a symmetry reduction of Pleba{\'n}ski's second heavenly equation, see \cite{Dunajski2004}.
In \cite{MorozovPavlov2017} it was shown that equations \eqref{eq:1} and \eqref{Pavlov_eq} are related by a
B\"acklund transformation. This transformation  produces, {\it inter alia},  non-local conservation laws for equation
\eqref{Pavlov_eq} from local conservation laws of equation \eqref{eq:1}, see \cite{LelitoMorozov2018}.

Equation~\eqref{eq:1} admits a Lax pair with a non-vanishing spectral
parameter~\cite{Zakh}, see also~\cite{Bur-Fer-Tsa},
and we use the scheme applied before in similar
situations: we expand this Lax pair in formal series with respect to the
spectral parameter and ``cut'' the result into two infinite-dimensional
coverings called the \emph{positive} and \emph{negative} ones~$\tau^+$
and~$\tau^-$, resp., see Section~\ref{sec:veron-web-equat}. Then, in
Section~\ref{sec:lie-algebr-nonl}, we give a full description of the Lie
algebras formed by nonlocal symmetries in these coverings. The arising
algebras are infinite-dimensional and possess quite an interesting structures,
to our opinion. In Section~\ref{sec:recurs-oper-mast}, we construct two
mutually inverse recursion operators and study their action on shadows of
nonlocal symmetries. An interesting feature of the VWE which distinguishes it
from other linearly degenerate equations is that the recursion operators
generate a new symmetry in the Whitney product of $\tau^+$ and~$\tau^-$. This
is a master-symmetry.

In the subsequent exposition, we use definitions and constructions from the
geometric theory of PDEs~\cite{AMS-book} and its nonlocal
version~\cite{Kras-Vin-Trends-1989}. A concise exposition of the necessary
material can be found in the preliminary parts of~\cite{H-K-M-V-2017}
or~\cite{B-K-M-V-2018}. Everywhere below we omit the proofs that are
accomplished by straightforward computations.

\section{The Veronese web equation and its coverings}
\label{sec:veron-web-equat}

Rewrite Equation~\eqref{eq:1} in the form
\begin{equation}\label{eq:2}
  u_{ty} = \frac{(\lambda + 1)u_tu_{xy} - u_yu_{tx}}{\lambda u_x}
\end{equation}
and choose the functions
\begin{equation*}
  x,\ y,\ t,\ u_{x^k},\ u_{x^ky^l},\ u_{x^kt^l},\quad k\geq0,\ l>0,
\end{equation*}
for internal coordinates on the infinite prolongation~$\mathcal{E}$ of VWE,
where
\begin{equation*}
  u_{x^iy^it^k} = \frac{\partial^{i+j+k} u}{\partial x^i\partial y^j\partial t^k}.
\end{equation*}
Then the total derivatives on~$\mathcal{E}$ acquire the form
\begin{equation}\label{eq:19}
\begin{array}{rcl}
  D_x&=&\dfrac{\partial}{\partial x} + \sum_k
       u_{x^{k+1}}\dfrac{\partial}{\partial u_{x^k}} +
       \sum_{k,l}\left(u_{x^{k+1}y^l}\dfrac{\partial}{\partial u_{x^ky^l}} +
       u_{x^{k+1}t^l}\dfrac{\partial}{\partial u_{x^kt^l}}\right),\\[5mm]
  D_y&=&\dfrac{\partial}{\partial y} + \sum_k
       u_{x^ky}\dfrac{\partial}{\partial u_{x^k}} +
       \sum_{k,l}\left(u_{x^ky^{l+1}}\dfrac{\partial}{\partial u_{x^ky^l}} +
       D_x^kD_t^{l-1}(R)\dfrac{\partial}{\partial u_{x^kt^l}}\right),\\[5mm]
  D_t&=&\dfrac{\partial}{\partial t} + \sum_k
       u_{x^kt}\dfrac{\partial}{\partial u_{x^k}} +
       \sum_{k,l}\left(D_x^kD_y^{l-1}(R)\dfrac{\partial}{\partial u_{x^ky^l}} +
       u_{x^kt^{l+1}}\dfrac{\partial}{\partial u_{x^kt^l}}\right),
\end{array}
\end{equation}
where~$R$ is the right-hand side of~\eqref{eq:2}.

The Lax pair for the VWE is
\begin{equation}
  \label{eq:3}
   w_{t}=\frac{\mu(\lambda+1)}{\lambda(\mu+1)} \frac{u_t}{u_x}
   w_x,\qquad w_{y}=\frac{\mu}{\lambda} \frac{u_y}{u_x} w_{x},
\end{equation}
where~$\mu\in\mathbb{R}$ is the spectral parameter. Expanding
$w=\sum\limits_i \mu^iw_i$, one obtains
\begin{equation}\label{eq:4}
  w_{i,t}=\frac{(\lambda+1)}{\lambda} \frac{u_t}{u_x}
  w_{i-1,x}-w_{i-1,t},\qquad
  w_{i,y}=\frac{1}{\lambda} \frac{u_y}{u_x} w_{i-1,x},
\end{equation}
where~$i\in\mathbb{Z}$. Setting~$w_i = 0$ for~$i<0$, we obtain the positive
covering; if we set~$w_i = 0$ for~$i>0$ the negative covering arises.

\subsection{The positive covering}
\label{sec:positive-covering}

The defining equations for the positive
covering~$\tau^+\colon \mathcal{E}^+ \to \mathcal{E}$ are
\begin{equation}\label{eq:5}
  q_{i,t}=\frac{(\lambda+1)}{\lambda} \frac{u_t}{u_x} q_{i-1,x}-q_{i-1,t},\qquad
  q_{i,y}=\frac{u_y q_{i-1,x}}{\lambda u_x},
\end{equation}
where~$i\geq 1$ and we formally set~$q_0 = x$. Nonlocal variables in the
covering~$\tau^+\colon \mathcal{E}^+ \to \mathcal{E}$ are~$q_{i,x}^k$,
$i\geq 1$, $k\geq 0$ (in particular,~$q_{i,x}^0 = q_i$), and the total
derivatives on~$\mathcal{E}^+$ are
\begin{equation}
  \label{eq:6}
  \tilde{D}_x = D_x + \sum_{i,k}q_{i,x}^{k+1}\frac{\partial}{\partial
    q_{i,x}^k},\quad \tilde{D}_y = D_y +
  \sum_{i,k}\tilde{D}_x^k(q_{i,y})\frac{\partial}{\partial
    q_{i,x}^k},\quad \tilde{D}_t = D_t +
  \sum_{i,k}\tilde{D}_x^k(q_{i,t})\frac{\partial}{\partial
    q_{i,x}^k},
\end{equation}
where~$q_{i,y}$ and~$q_{i,t}$ are given by Equations~\eqref{eq:5}.

Consider the tower of coverings
\begin{equation*}
  \xymatrix{
    \mathcal{E}^+\ar[r]& \dots\ar[r]&\mathcal{E}_{i+1}^+\ar[r]&
    \mathcal{E}_i^+\ar[r]& \dots\ar[r]& \mathcal{E}_0^+ =\mathcal{E}\rlap{,}
  }
\end{equation*}
where nonlocal variables in~$\mathcal{E}_i^+$ are~$q_{\alpha,x}^k$, $1\leq
\alpha \leq i$, $k\geq 0$.

\begin{proposition}\label{sec:positive-covering-prop-1}
  The $2$-forms
  \begin{equation*}
    \omega_{i+1}^k = \left(\tilde{D}_x^k(q_{i+1,y})\,dy +
      \tilde{D}_x^k(q_{i+1,t})\,dt\right)\wedge\,dx
  \end{equation*}
  are linearly independent $2$-component conservation laws on~$\mathcal{E}_i^+$.
\end{proposition}

When proving this statement, as well as
Proposition~\ref{sec:negative-covering-prop-1}, we use the following fact
(see~\cite{Kras-2015}): Let~$\mathcal{E}$ be a differentially connected
equation (i.e., an equation such that the only functions on~$\mathcal{E}$
annihilated by all total derivatives are constants) and let~$\tau_\Omega\colon
\mathcal{E}_\Omega \to \mathcal{E}$ be an Abelian covering associated with a
system of conservation laws~$\Omega = \{\omega_1,\dots,\omega_s\}$. Then these
conservation laws are linearly independent if and only if the covering equation
is differentially connected as well.

\begin{proof}
  The proof uses a double induction: on~$i$ and on~$k$ for each~$i$. Obviously,
  VWE is differentially connected; denote by~$\mathcal{K}_s$ the space of
  functions that are annihilated by~$\tilde{D}_x$, $\tilde{D}_y$,
  $\tilde{D}_t$ and have jet order~$\leq s$. Let~$f\in \mathcal{K}_s$.

  \begin{case}[$i=1$, $k=0$]
    We have
    \begin{equation*}
      \tilde{D}_y = D_y +
      \frac{1}{\lambda}\frac{u_y}{u_x}\frac{\partial}{\partial q_1}
    \end{equation*}
    in this case, from where it follows that~$s=0$, i.e., $f =
    f(x,y,t,u,q_1)$. Hence,~$f$ must be invariant w.r.t.
    \begin{equation*}
      \frac{\partial}{\partial y} + u_y\frac{\partial}{\partial u} +
      \frac{1}{\lambda}\frac{u_y}{u_x}\frac{\partial}{\partial q_1}.
    \end{equation*}
    But $f$ is independent of~$u_x$; consequently, it does not depend
    on~$q_1$.
  \end{case}

  \begin{case}[$i=1$, $k>0$]\label{sec:positive-covering-step-2}
    One has
    \begin{equation*}
      \tilde{D}_y = D_y + \frac{1}{\lambda}\sum_{\alpha=0}^k
      \tilde{D}_x^\alpha\left(\frac{u_y}{u_x}\right) \frac{\partial}{\partial
        q_{1,x}^\alpha}.
    \end{equation*}
    The maximal jet order of the nonlocal summand is~$k+1$ and the variables
    of this order are~$u_{x^{k+1}}$ and~$u_{x^ky}$. Using
    expression~\eqref{eq:19} for~$\tilde{D}_y$, we see that
    \begin{equation*}
      Z_1 = \Big[\frac{\partial}{\partial u_{x^{k+1}}}, \tilde{D}_y\Big]=
      -\frac{\lambda +1}{\lambda}\frac{u_tu_{xy} -
        u_yu_{xt}}{u_x^2}\frac{\partial}{\partial u_{x^kt}} -
      \frac{1}{\lambda}\frac{u_y}{u_x^2}\frac{\partial}{\partial q_{1,x}^k}.\\
    \end{equation*}
    On the other hand, since the coefficients of~$\tilde{D}_y$ do not depend
    on~$u_{y^l}$, $l>1$, the function~$f\in\mathcal{K}_s$ cannot depend
    on~~$u_{y^l}$, $l>0$. Consequently, it must be invariant w.r.t.
    \begin{equation*}
      Z_2 =\Big[\frac{\partial}{\partial u_y},Z_1\Big] = \frac{\lambda
        +1}{\lambda}\frac{u_{xt}}{u_x^2}\frac{\partial}{\partial u_{x^kt}} -
      \frac{1}{\lambda}\frac{1}{u_x^2}\frac{\partial}{\partial q_{1,x}^k}.
    \end{equation*}
    But the field~$\partial/\partial q_{1,x}^k$ is a linear combination
    of~$Z_1$ and~$Z_2$; hence,~$f$ does not depend on~$q_{1,x}^k$.
  \end{case}

  \begin{case}[$i>1$, $k=0$]
    From Equations~\eqref{eq:6} and~\eqref{eq:5} we see that
    \begin{multline*}
      \tilde{D}_y = D_y +
      \frac{1}{\lambda}\sum_{\beta=1}^i
      \sum_{\alpha=0}^{k_\beta}\tilde{D}_x^\alpha\left(\frac{u_y
          q_{\beta-1,x}}{u_x}\right)\frac{\partial}{\partial
        q_{\beta,x}^\alpha}\\= D_y +
      \frac{1}{\lambda}\sum_{\beta=1}^i
      \sum_{\alpha=0}^{k_\beta}\left(\frac{u_y
          }{u_x}q_{\beta-1,x}^{\alpha+1} +
          \alpha\tilde{D}_x\left(\frac{u_y}{u_x}\right)q_{\beta-1,x}^\alpha +
            \dots\right)\frac{\partial}{\partial
        q_{\beta,x}^\alpha}
    \end{multline*}
    This means that the inequalities
    \begin{equation}\label{eq:20}
      k_1 > k_2 > \dots > k_i
    \end{equation}
    hold. Consequently, the maximal jet order of the
    coefficients appears at the term
    \begin{equation*}
      \frac{1}{\lambda}\tilde{D}_x^{k_1}\left(\frac{u_y}{u_x}q_{1,x}\right)
      \frac{\partial}{\partial q_{1,x}^{k_1}},
    \end{equation*}
    which means that repeating the reasoning of
    Step~\ref{sec:positive-covering-step-2} one can prove
    that~$f\in \mathcal{K}_s$ is independent of the
    variables~$q_{1,x}^{k_1},\dots,q_{1,x}^{k_2+1}$, i.e.,~$k_1=k_2$, which is
    impossible by~\eqref{eq:20}.
  \end{case}

  \begin{case}[$i>1$, $k>0$]
    The proof at this step repeats literary the one accomplished at
    Step~\ref{sec:positive-covering-step-2} for~$i=1$.
  \end{case}
  The result is proved.
\end{proof}

\subsection{The negative covering}
\label{sec:negative-covering}

The negative covering~$\tau^-\colon \mathcal{E}^- \to \mathcal{E}$ is defined
by the system
\begin{equation}\label{eq:7}
  r_{i,t}=\frac{(\lambda+1)u_t r_{i-1,y}}{u_y}-r_{i-1,t},\qquad
  r_{i,x}=\frac{\lambda u_x r_{i-1,y}}{u_y},
\end{equation}
where~$i\geq1$ and~$r_0 = y$. Nonlocal variables in~$\mathcal{E}^-$
are~$r_{i,y}^k$, $i\geq 1$, $k\geq 0$ (in particular, $r_{i,x}^0 = r_i$),
while the total derivatives take the form
\begin{equation}\label{eq:8}
  \tilde{D}_x = D_x + \sum_{i,k}\tilde{D}_y^k(r_{i,x})\frac{\partial}{\partial
    r_{i,y}^k},\quad  \tilde{D}_y = D_y + \sum_{i,k}
  r_{i,y}^{k+1}\frac{\partial}{\partial r_{i,y}^k}, \quad \tilde{D}_t = D_t +
  \sum_{i,k} \tilde{D}_y^k(r_{i,t})\frac{\partial}{\partial r_{i,y}^k},
\end{equation}
where~$r_{i,x}$, $r_{i,t}$ are given by Equations~\eqref{eq:7}.

Similar to the positive case, we consider the tower
\begin{equation*}
  \xymatrix{
    \mathcal{E}^-\ar[r]& \dots\ar[r]&\mathcal{E}_{i-1}^-\ar[r]&
    \mathcal{E}_i^-\ar[r]& \dots\ar[r]& \mathcal{E}_0^- =\mathcal{E}\rlap{,}
  }
\end{equation*}
where nonlocal variables in~$\mathcal{E}_i^-$ are~$r_{\alpha,y}^k$, $1\leq
\alpha \leq i$, $k\geq 0$.

\begin{proposition}\label{sec:negative-covering-prop-1}
  The $2$-forms
  \begin{equation*}
    \theta_{i+1}^k = \left(\tilde{D}_y^k(r_{i+1,x})\,dx +
      \tilde{D}_y^k(r_{i+1,t})\,dt\right)\wedge\,dy
  \end{equation*}
  are linearly independent $2$-component conservation laws on~$\mathcal{E}_i^-$.
\end{proposition}

\begin{proof}
  The proof of this statement is similar to the one of
  Proposition~\ref{sec:positive-covering-prop-1}, but we must work with the
  field~$\tilde{D}_x$ instead of~$\tilde{D}_y$.
\end{proof}

\section{Lie algebras of nonlocal symmetries}
\label{sec:lie-algebr-nonl}

Local symmetries of~$\mathcal{E}$ are solutions to the linearization
\begin{equation}
  \ell_{\mathcal{E}}(\phi) =0
  \label{linearization_eq}
\end{equation}
of Equation~\eqref{eq:1}, where
\begin{equation}
  \label{eq:9}
  \ell_{\mathcal{E}} =
  \big(u_yD_xD_t + u_{xt}D_y\big)
  +\lambda\big(u_xD_tD_y + u_{ty}D_x\big) -
  (\lambda+1)\big(u_tD_xD_y + u_{xy}D_t\big)
  ;
\end{equation}
they form a Lie algebra w.r.t.\ the Jacobi bracket~$\{\cdot\,,\cdot\}$ denoted
by~$\sym(\mathcal{E})$. The corresponding vector field on~$\mathcal{E}$ is the
evolution derivation
\begin{equation*}
  \Ev_\phi = \sum D_\sigma(\phi)\frac{\partial}{\partial u_\sigma},
\end{equation*}
where summation is done over all the internal coordinates~$u_\sigma$
on~$\mathcal{E}$.

Direct computations lead to the following
\begin{proposition}
  \label{sec:lie-algebr-loc}
  The algebra~$\sym(\mathcal{E})$ is spanned by the elements
  \begin{equation*}
    \phi_1(T)=Tu_t,\quad\phi_2(X)=Xu_x,\quad\phi_3(Y)=Yu_y,\quad \phi_4(U)=U,
  \end{equation*}
  where~$T=T(t)$\textup{,} $X=X(x)$\textup{,} $Y=Y(y)$\textup{,} $U=U(u)$ are
  arbitrary smooth functions. The non-zero commutators are
  \begin{align*}
    &\{\phi_1(T),\phi_1(\tilde{T})\} = \phi_1([\tilde{T},T]),
    &&\{\phi_2(X),\phi_2(\tilde{X})\} = \phi_2([\tilde{X},X]),  \\
    &\{\phi_3(Y),\phi_3(\tilde{Y})\} = \phi_3([\tilde{Y},Y]),
    &&\{\phi_4(U),\phi_4(\tilde{U})\} = \phi_4([U,\tilde{U}]),
  \end{align*}
  where~$[Z,\tilde{Z}]$
  denotes~$Z\partial \tilde{Z}/\partial z - \tilde{Z}\partial Z /\partial z$
  for any functions~$Z$ and~$\tilde{Z}$ in~$z$.
\end{proposition}

\subsection{The algebra $\sym_{\tau^+}(\mathcal{E})$}
\label{sec:algebra-sym_t}

To find the Lie algebra~$\sym_{\tau^+}(\mathcal{E})$ of nonlocal symmetries
in the positive covering~$\tau^+$, one needs to solve the following system:
\begin{equation}\label{eq:10}
  \begin{aligned}
    \tilde{\ell}_{\mathcal{E}}(\phi)&=0,\\
    \tilde{D}_t(\phi^i)& =
    \frac{(\lambda+1)}{\lambda}\left(\frac{u_x\tilde{D}_t(\phi) -
        u_t\tilde{D}_x(\phi)
      }{u_x^2}q_{i-1,x} + \frac{u_t}{u_x}\tilde{D}_x(\phi^{i-1}) \right) -
    \tilde{D}_t(\phi^{i-1}),\\[.5mm]
    \tilde{D}_y(\phi^i)& = \frac{1}{\lambda}\left(\frac{u_x\tilde{D}_y(\phi) -
      u_y\tilde{D}_x(\phi)}{u_x^2}q_{i-1,x} +
      \frac{u_y }{u_x}\tilde{D}_x(\phi^{i-1})\right),
  \end{aligned}
\end{equation}
where~$\tilde{\ell}_{\mathcal{E}}$ denotes the natural lift of the
operator~\eqref{eq:9} from~$\mathcal{E}$ to~$\mathcal{E}^+$. Solutions
of~\eqref{eq:10} are denoted by
\begin{equation*}
  \Phi = [\phi,\phi^1,\dots,\phi^i,\dots],
\end{equation*}
and to any such a~$\Phi$ there corresponds the vector field
\begin{equation*}
  S_\Phi = \tilde{\Ev}_\phi +
  \sum_{i,k}\tilde{D}_x^k(\phi^i)\frac{\partial}{\partial q_{i,x}^k}
\end{equation*}
on~$\mathcal{E}^+$. Solutions of the first equation in~\eqref{eq:10} are called
nonlocal $\tau^+$-shadows. In particular, local symmetries can be considered
as shadows in any covering. If~$\phi$ is a shadow and there exists a nonlocal
symmetry~$\Phi = [\phi,\phi^1,\dots]$ then we say that this shadow
\emph{lifts} to the covering. Nonlocal symmetries~$\Phi$ with~$\phi=0$ (i.e.,
with trivial shadows) are called \emph{invisible}. Given~$\Phi$
and~$\tilde{\Phi}$, one can define their bracket by
\begin{equation*}
  \{\Phi,\tilde{\Phi}\} = S_\Phi(\tilde{\Phi}) - S_{\tilde{\Phi}}(\Phi),
\end{equation*}
where the action is component-wise.

Consider the vector field
\begin{equation}\label{eq:11}
  \mathcal{X}= \sum\limits_{i=0}^{\infty}(i+1)q_{i+1}\frac{\partial}{\partial q_i},
\end{equation}
(recall that~$q_0 = x$) and set
\begin{equation}\label{eq:12}
  P_0(X)=X, \qquad P_j(X)=\frac{1}{j}\mathcal{X}(P_{j-1}(X)), \quad j\geq 1.
\end{equation}

\begin{proposition}
  \label{sec:algebra-sym_t-prop-loc-lift}
  All the local symmetries of the VWE can be lifted to the positive covering.
\end{proposition}

\begin{proof}
  Denote
  by~$\Phi_\alpha = [\phi_\alpha,\phi_\alpha^1,\dots,\phi_\alpha^i,\dots]$,
  $\alpha = 1,\dots,4$, the lifts to be constructed and set
  \begin{equation*}
    \phi_1^i(T)=T q_{i,t},\quad \phi_2^i(X)=Xq_{i,x}-P_i(X),\quad
    \phi_3^i(Y)=Yq_{i,y},\quad \phi_4^i(U)=0.
  \end{equation*}
  A direct computation shows that the functions~$\Phi_1 = \Phi_1(T)$,
  $\Phi_2 = \Phi_2(X)$, $\Phi_3 = \Phi_3(Y)$, and~$\Phi_4 = \Phi_4(U)$ are the
  desired lifts.
\end{proof}

We now construct two series of $\tau^+$-nonlocal symmetries.

The first one, denoted
by~$\Psi_k^+ = [\psi_k^+, \psi_k^{+,1}, \dots, \psi_k^{+,i}, \dots]$,
$k\geq 2$, arises as follows. The symmetries~$\Psi_2^+$, $\Psi_3^+$,
and~$\Psi_4^+$ are introduced ``by hand'':
\begin{align*}
  \psi_2^+&=(2q_2-q_1(q_{1,x}-1))u_x,\\
  \psi_2^{+,i}&=-(i+2)q_{i+2}-(i+1)q_{i+1}+2q_2q_{i,x}+q_1\left(
                q_{i+1,x}-(q_{1,x}-1)q_{i,x} \right)\\
                &+ \frac{1}{\lambda} ((i+1)q_{i+1}+iq_i-q_1 q_{i,x});\\[2mm]
  \psi_3^+&=(3q_3-2q_2q_{1,x}-q_1(q_{2,x}-q_{1,x}^2+1))u_x,\\
  \psi_3^{+,i}& =-(i+3)q_{i+3}+(i+1)q_{i+1}+3q_3q_{i,x}
                +2q_2(q_{i+1,x}-q_{1,x}q_{i,x})\\
          &+q_1\left( q_{i+2,x}-q_{1,x}q_{i+1,x}-(q_{2,x}-q_{1,x}^2+1)q_{i,x} \right)\\[2mm]
          &+ \frac{1}{\lambda}
            ((i+2)q_{i+2}-iq_i-2q_2q_{i,x}-q_1(q_{i+1,x}-q_{1,x}q_{i,x}));\\[2mm ]
  \psi_4^+&=(4q_4-3q_3q_{1,x}-2q_2(q_{2,x}-q_{1,x}^2)
            -q_1(q_{3,x}-2q_{1,x}q_{2,x}+q_{1,x}^3-1))u_x, \\
  \psi_4^{+,i}&=-(i+4)q_{i+4}-(i+1)q_{i+1}+4q_4q_{i,x}
                +3q_3(q_{i+1,x}-q_{1,x}q_{i,x})\\
          &+2q_2\left(q_{i+2,x}-q_{1,x}q_{i+1,x}
            -(q_{2,x}-q_{1,x}^2)q_{i,x} \right)\\
          &+q_1(q_{i+3,x}-q_{1,x}q_{i+2,x}-(q_{2,x}-q_{1,x}^2)q_{i+1,x}
            -(q_{3,x}-2q_{1,x}q_{2,x}+q_{1,x}^3-1)q_{i,x})\\
          &+\frac{1}{\lambda}((i+3)q_{i+3}+iq_i-
            3q_3q_{i,x}-2q_2(q_{i+1,x}-q_{1,x}q_{i,x})\\
            &-q_1(q_{i+2,x}-q_{1,x}q_{i+1,x}-(q_{2,x}-q_{1,x}^2)q_{i,x})).
\end{align*}
For~$k>4$, we set
\begin{multline*}
  \Psi_k^+=\frac{1}{k-4} \Big( \left\{\Psi_{k-2}^+,\Psi_2^+ \right\}
    -(k-3)\Psi_{k-1}^+ +(-1)^k \Psi_{3}^+ \\+ \frac{1}{\lambda} \left(
      (k-4)\Psi_{k-1}^+ + (k-3)\Psi_{k-2}^+ - (-1)^k \Psi_{2}^+\right)
  \Big).
\end{multline*}

Now, the second
series~$\Xi_k^+(X)=[\xi_k^+(X), \xi_k^{+,1}(X),\dots,\xi_k^{+,i}(X),\dots]$,
$k\geq 1$, is defined by the relations
\begin{align*}
  \xi_1^+(X)&=(Xq_{1,x}-X_xq_1)u_x,\\
  \xi_1^{+,i}(X)&=X(q_{1,x}q_{i,x}-q_{i+1,x})-X_xq_1q_{i,x}+P_{i+1}(X)+
                \frac{1}{\lambda} (Xq_{i,x}-P_i(X));\\[2mm]
  \xi_2^+(X)&=\left(
            X(q_{2,x}-q_{1,x}^2)+X_x(q_1q_{1,x}-q_2)-\frac{1}{2}X_{xx}q_1^2
            \right)u_x, \\
  \xi_2^{+,i}(X)&=X(q_{1,x}q_{i+1,x}+(q_{2,x}-q_{1,x}^2)q_{i,x}
                -q_{i+2,x})-X_x(q_2q_{i,x}+q_1(q_{i+1,x}-q_{1,x}q_{i,x}))\\
    &- \frac{1}{2}X_{xx}q_1^2q_{i,x} +P_{i+2}(X)+
      \frac{1}{\lambda}(X(q_{i+1,x}-q_{1,x}q_{i,x})+X_xq_1q_{i,x}-P_{i+1}(X)),
\end{align*}
where the functions~$P_i(X)$ are given by relations~\eqref{eq:12}, and
\begin{equation*}
  \Xi_k^+(X)=\frac{1}{k-2} \Big(
    \left\{\Xi_{k-2}^+(X),\Psi_2^+ \right\} -\Xi_{k-1}^+
    (X) + \frac{1}{\lambda}(k-3) \left(\Xi_{k-1}^+
      (X)+\Xi_{k-2}^+ (X) \right) \Big)
\end{equation*}
for $k\geq 3$.

Finally, invisible symmetries in~$\tau^+$ are
\begin{equation*}
  \Phi_k^{\inv}(X)=[\underbrace{0, \ldots, 0}_{\rm {\it k}-times}, P_0(X),
  P_1(X),\dots],
\end{equation*}
where~$k\geq 1$ and $P_i(X)$ are given by~\eqref{eq:12}, as above.

To describe the Lie algebra structure in~$\sym_{\tau^+}(\mathcal{E})$, it is
convenient to relabel the above introduced symmetries. Namely, we change the
generators of~$\sym_{\tau^+}(\mathcal{E})$ as follows:
\begin{equation*}
  \Phi_2(X)\mapsto -\Xi_0^+(X),\qquad \Phi_k^{\inv}(X)
  \mapsto \Xi_{-k}^+(X),\ k\geq1,
\end{equation*}
and
\begin{equation*}
  \Psi_k^+\mapsto (-1)^{k+1}\Psi_k^+,\ k\geq 2,\qquad \Xi_k^+(X)\mapsto
  (-1)^k\Xi_k^+(X), \ k\in\mathbb{Z}.
\end{equation*}

\begin{proposition}
  \label{sec:algebra-sym_t-brackets}
  In the new basis\textup{,} the Lie algebra structure
  of~$\sym_{\tau^+}(\mathcal{E})$ is given by the brackets
  \begin{equation*}
    \{\Psi_i^+,\Psi_j^+\} = (j-i)\Big(\Psi_{i+j}^+ + \frac{1}{\lambda}
    \Psi_{i+j-1}^+\Big) - (j-1)\Big(\Psi_{j+1}^+ +
    \frac{1}{\lambda}\Psi_j^+\Big) + (i-1)\Big(\Psi_{i+1}^+ +
    \frac{1}{\lambda}\Psi_i^+\Big)
  \end{equation*}
  for all~$j>i\geq 2$. One also has
  \begin{equation*}
    \{\Xi_i^+(X),\Xi_j^+(\tilde{X})\} =
    \begin{cases}
      \Xi_{i+j}^+([X,\tilde{X}]),& i,j\leq 0 \text{ or } i<0,j>0, i+j>0,\\
      \Xi_{i+j}^+([X,\tilde{X}]) +
      \dfrac{1}{\lambda}\Xi_{i+j-1}^+([X,\tilde{X}]),& \text{otherwise,}
    \end{cases}
  \end{equation*}
  $i$, $j\in \mathbb{Z}$, and
  \begin{equation*}
    \{\Psi_i^+,\Xi_j^+(X)\} =
    \begin{cases}
      j\big(\Xi_{i+j}^+(X) - \Xi_{j+1}^+(X)\big) +
      \dfrac{1}{\lambda}\big((j-1)\Xi_{i+j-1}^+(X) - \Xi_j^+(X)\big),&j\geq
      1,\\[3mm]
      j\Big(\Xi_{i+j}^+(X) - \Xi_{j+1}^+(X) -
      \dfrac{1}{\lambda}\Xi_j^+(x)\Big),&j<1,
      i+j>0,\\[3mm]
      j\Big(\Xi_{i+j}^+(X) - \Xi_{j+1}^+(X) -
      \dfrac{1}{\lambda}\big(\Xi_{i+j-1}^+(X)
      - \Xi_j^+(X)\big)\Big),&\text{ otherwise},
    \end{cases}
  \end{equation*}
  $i\geq 2$\textup{,} $j\in\mathbb{Z}$. All the other commutators vanish.
\end{proposition}

\subsection{The algebra $\sym_{\tau^-}(\mathcal{E})$}
\label{sec:algebra-sym_t-mathc}

Computations here go along the same lines as in
Subsection~\ref{sec:algebra-sym_t} and we use similar notation below. The
defining equations are
\begin{equation}\label{eq:13}
  \begin{aligned}
    \tilde{\ell}_{\mathcal{E}}(\phi)&=0,\\
    \tilde{D}_t(\phi^i)& =(\lambda+1)\left(\frac{u_y\tilde{D}_t(\phi) -
        u_t\tilde{D}_y(\phi)}{u_y^2}r_{i-1,y} +
      \frac{u_t }{u_y} \tilde{D}_y(\phi^{i-1})\right) - \tilde{D}_t(\phi^{i-1})
    ,\\[.5mm]
    \tilde{D}_x(\phi^i)& = \lambda\left(\frac{u_y\tilde{D}_x(\phi)
        -u_x\tilde{D}_y(\phi)}{u_y^2}r_{i-1,y} +
      \frac{u_x}{u_y}\tilde{D}_y(\phi^{i-1})\right),
  \end{aligned}
\end{equation}
where ``tilde'' marks operators on~$\mathcal{E}^-$. A solution $\Phi =[\phi,
\phi^1,\dots, \phi^i, \dots]$ of~\eqref{eq:13} corresponds to the vector field
\begin{equation*}
  S_\Phi = \tilde{\Ev}_\Phi +
  \sum_{i,k}\tilde{D}_y^k(\phi^i)\frac{\partial}{\partial r_{i,y}^k}
\end{equation*}
on~$\mathcal{E}^-$ and the bracket $\{\Phi,\tilde{\Phi}\} =
S_\Phi(\tilde{\Phi}) - S_{\tilde{\Phi}}(\Phi)$ is defined for such solutions.

To proceed with further constructions, we will introduce the vector field
\begin{equation}\label{eq:15}
  \mathcal{Y}= \sum_{i=0}^{\infty}(i+1)r_{i+1}\frac{\partial}{\partial r_i},
\end{equation}
and the quantities $Q_j$, $j=0,1\dots$ defined as follows:
\begin{equation}\label{eq:14}
  Q_0(Y)=Y, \quad Q_j(Y)=\frac{1}{j}\mathcal{Y}(Q_{j-1}(Y)), \qquad j\geq 1,
\end{equation}
(recall that $r_0 = y$).

\begin{proposition}
  \label{sec:algebra-sym_t-mathc-loc}
  All the local symmetries of the WVE can be lifted to the negative covering.
\end{proposition}

\begin{proof}
  Denote the lifts by~$\Phi_\alpha = [\phi_\alpha,\phi_\alpha^1,\dots,
  \phi_\alpha^i, \dots]$, $\alpha=1,\dots,4$, and set
  \begin{equation*}
    \phi_1^i(T)=T r_{i,t},\quad
    \phi_2^i(X)=Xr_{i,x},\quad
    \phi_3^i(Y)=Yr_{i,y}-Q_i(Y),\quad
    \phi_4^i(U) = 0.
  \end{equation*}
  The rest of proof is a straightforward check.
\end{proof}

Let us now construct, similar to the positive case, two series of nonlocal
symmetries. The first one~$\Psi_k^- = [\psi_k^-,\psi_k^{-,1}, \dots,
\psi_k^{-,i}, \dots]$, $k\geq 2$, is defined as follows. For~$k=2$, $3$, $4$
we set
\begin{align*}
  \psi_2^-&=(2r_2-r_1(r_{1,y}-1))u_y,\\
  \psi_2^{-,i}&=-(i+2)r_{i+2} - (i+1)r_{i+1}+2r_2r_{i,y} + r_1\left(
                r_{i+1,y}-(r_{1,y}-1)r_{i,y} \right)\\
          &+ \lambda ((i+1)r_{i+1}+ir_i-r_1 r_{i,y});\\[2mm]
  \psi_3^-&=(3r_3-2r_2r_{1,y}-r_1(r_{2,y}-r_{1,y}^2+1))u_y,\\
  \psi_3^{-,i}&=-(i+3)r_{i+3}+(i+1)r_{i+1}+3r_3r_{i,y}+2r_2(r_{i+1,y}-r_{1,y}r_{i,y})\\
  &+r_1\left( r_{i+2,y}-r_{1,y}r_{i+1,y}-(r_{2,y}-r_{1,y}^2+1)r_{i,y}
    \right)\\
  &+ \lambda
    ((i+2)r_{i+2}-ir_i-2r_2r_{i,y}-r_1(r_{i+1,y}-r_{1,y}r_{i,y}));\\[2mm]
  \psi_4^-&=(4r_4-3r_3r_{1,y}-2r_2(r_{2,y}-r_{1,y}^2)
            -r_1(r_{3,y}-2r_{1,y}r_{2,y}+r_{1,y}^3-1))u_y,\\
  \psi_4^{-,i}&=-(i+4)r_{i+4}-(i+1)r_{i+1}+4r_4r_{i,y}+3r_3(r_{i+1,y}-r_{1,y}r_{i,y})\\
  &+2r_2\left( r_{i+2,y}-r_{1,y}r_{i+1,y}-(r_{2,y}-r_{1,y}^2)r_{i,y} \right)\\
    &+r_1(r_{i+3,y}-r_{1,y}r_{i+2,y}-(r_{2,y}-r_{1,y}^2)r_{i+1,y}-(r_{3,y}-2r_{1,y}r_{2,y}+r_{1,y}^3-1)r_{i,y})\\
    &+\lambda((i+3)r_{i+3}+ir_i-3r_3r_{i,y}-2r_2(r_{i+1,y}-r_{1,y}r_{i,y})-r_1(r_{i+2,y}-r_{1,y}r_{i+1,y}-(r_{2,y}-r_{1,y}^2)r_{i,y})).
\end{align*}
For $k>4$ we define
\begin{multline*}
  \Psi_k^-=\frac{1}{k-4} \big(
  \left\{\Psi_{k-2}^-,\Psi_2^- \right\}
  -(k-3)\Psi_{k-1}^- +(-1)^k \Psi_{3}^- \\
  + \lambda \left( (k-4)\Psi_{k-1}^- + (k-3)\Psi_{k-2}^- -
    (-1)^k \Psi_{2}^-\right) \big).
\end{multline*}

Introduce the second
series~$\Xi_k^-(Y) = [\xi_k^-(Y),\xi_k^{-,1}(Y),\dots,\xi_k^{-,i}(Y),\dots]$ now by
\begin{align*}
  \xi_1^-(Y)&=(Yr_{1,y}-Y_yr_1)u_y,\\
  \xi_1^{-,i}(Y)&=Y(r_{1,y}r_{i,y}-r_{i+1,y})-Y_yr_1r_{i,y}+Q_{i+1}(Y)+ \lambda
                (Yr_{i,y}-Q_i(Y));\\[2mm]
  \xi_2^-(Y)&=\left( Y(r_{2,y}-r_{1,y}^2)+Y_y(r_1r_{1,y}-r_2)
            -\frac{1}{2}Y_{yy}r_1^2 \right)u_y,\\
  \xi_2^{-,i}(Y)&=Y(r_{1,y}r_{i+1,y}+(r_{2,y}-r_{1,y}^2)r_{i,y}
                -r_{i+2,y})-Y_y(r_2r_{i,y}+r_1(r_{i+1,y}-r_{1,y}r_{i,y}))\\
          &- \frac{1}{2}Y_{yy}r_1^2r_{i,y} +Q_{i+2}(Y)+
            \lambda(Y(r_{i+1,y}-r_{1,y}r_{i,y})+Y_yr_1r_{i,y}-Q_{i+1}(Y))
\end{align*}
and for~$k\geq3$
\begin{equation*}
  \Xi_k^-(Y)=\frac{1}{k-2} \big(
  \left\{\Xi_{k-2}^-(Y),\Psi_2 \right\} -\Xi_{k-1}^-
  (Y) + \lambda(k-3) \left(\Xi_{k-1}^-(Y)+\Xi_{k-2}^-(Y)
  \right) \big).
\end{equation*}
Invisible symmetries in~$\tau^-$ have the form
\begin{equation*}
  \Phi_k^{\inv}(Y)=[\underbrace{0, \ldots, 0}_{\rm {\it k}-times},  Q_0(Y),
  Q_1(Y), \dots]
\end{equation*}
where~$k\geq 1$ and $Q_i(Y)$ are given by~\eqref{eq:14}.

We again relabel the generators by
\begin{equation*}
  \Phi_3(Y)\mapsto -\Xi_0^-(Y), \qquad \Phi_k^{\inv}(Y) \mapsto
  \Xi_{-k}^-(Y),\ k\geq 1,
\end{equation*}
and
\begin{equation*}
  \Psi_k^-\mapsto (-1)^{k+1}\Psi_k^-,\ k\geq 2,\qquad \Xi_k^-(Y)\mapsto
  (-1)^k\Xi_k^-(Y),\ k\in\mathbb{Z}.
\end{equation*}
Then the following statement holds:
\begin{proposition}
  \label{sec:algebra-sym_t-mathc-Lie-str}
  The above defined generators enjoy the following relations\textup{:}
  \begin{equation*}
    \{\Psi_i^-,\Psi_j^-\} = (j-i)(\Psi_{i+j}^- + \lambda\Psi_{i+j-1}^-) -
    (j-1)(\Psi_{j+1}^- + \lambda \Psi_j^-) + (i-1)(\Psi_{i+1}^- + \lambda\,\Psi_i^-)
  \end{equation*}
  for $j>i\geq 2$,
  \begin{equation*}
    \{\Xi_i^-(Y),\Xi_j^-(\tilde{Y})\} =
    \begin{cases}
      \Xi_{i+j}^-([Y,\tilde{Y}]),&i\leq0,j>0,i+j>0\text{ or }j\leq0,\\[2mm]
      \Xi_{i+j}^-([Y,\tilde{Y}]) +
      \lambda\Xi_{i+j-1}^-([Y,\tilde{Y}]),&\text{otherwise},
    \end{cases}
  \end{equation*}
  for all~$i<j\in\mathbb{Z}$\textup{,} and
  \begin{equation*}
    \{\Psi_i^-,\Xi_j^-(Y)\} =
    \begin{cases}
      j\big(\Xi_{i+j}^-(Y) - \Xi_{j+1}^-(Y)\big) +
      \lambda\big((j-1)\Xi_{i+j-1}^-(Y)
      - \Xi_j^-(Y)\big),&j\geq1\\[2mm]
      j\big(\Xi_{i+j}^-(Y) -\Xi_{j+1}^-(Y) -
      \lambda\Xi_j^-(Y)\big),&j<1,i+j>0,\\[2mm]
      j\big(\Xi_{i+j}^-(Y) - \Xi_{j+1}^-(Y) + \lambda(\Xi_{i+j-1}^-(Y) -
      \Xi_j^-(Y))\big),&\text{otherwise},
    \end{cases}
  \end{equation*}
  for $j\in \mathbb{Z}$\textup{,} $i\geq 2$.
\end{proposition}

\section{Recursion operators and a master-symmetry}
\label{sec:recurs-oper-mast}

According to the general theory, see~\cite{Marvan-1995}, recursion operators
for symmetries arise as B\"{a}cklund auto-transformations of the tangent
space~$\mathcal{T}\mathcal{E}$, cf.~\cite{KVV}. In the case of VWE, this BT is
\begin{equation}
  \label{eq:16}
  \begin{array}{lcl}
    u_x\tilde{D}_t(\zeta)
    &=& \displaystyle{
        u_{tx}\,\zeta
        -u_x\tilde{D}_t(\eta)
        +\frac{\lambda+1}{\lambda}\,u_t\,\tilde{D}_x(\eta)
        -\frac{1}{\lambda}\,u_{tx}\,\eta,
        }
    \\[4mm]
    u_x\tilde{D}_y(\zeta)
    &=& \displaystyle{
        u_{xy}\,\zeta +\frac{1}{\lambda}\,u_y\,\tilde{D}_x(\eta)
        -\frac{1}{\lambda}\,u_{xy}\,\eta.
        }
\end{array}
\end{equation}

\begin{proposition}
  \label{sec:recurs-oper-mast-rec-op}
  Let~$\eta$ be a $\tau^\pm$-shadow. Then~$\zeta=\mathcal{R}_+(\eta)$ obtained
  as a solution of~\eqref{eq:16} is a shadow as well. Vice versa\textup{,} if
  $\eta$ is a shadow the~$\zeta = \mathcal{R}_-(\eta)$ obtained in the same
  way is a shadow too.
\end{proposition}

\begin{proof}
  To construct a recursion operator for Equation~\eqref{eq:1} we use the
  techniques of \cite{Sergyeyev2015},
  cf.~\cite{MalykhNutkuSheftel2004,MarvanSergyeyev2012,MorozovSergyeyev2014,KruglikovMorozov2015,B-K-M-V-2018}
  also.  We find a shadow for VWE in the covering~\eqref{eq:3}. It is of the
  form $s=H(w)\,u_x\,w_x^{-1}$, where $H$ is an arbitrary function in
  $w$. Since System~\eqref{eq:3} is invariant with respect to the
  transformation $w \mapsto H(w)$, we put, without loss of generality,
  $s=u_x\,w_x^{-1}$. Differentiation of~\eqref{eq:3} by~$x$ and
  substitution $q_x =u_x\,s^{-1}$ gives another covering
\begin{equation}
  \label{Vwe_s_covering}
  s_t =
  \frac{\mu\,(\lambda+1)\,u_t}{\lambda\,(\mu+1)\,u_x}\,s_x+
  \frac{(\lambda-\mu)\,u_{tx}}{\lambda\,u_x}\,s,
  \qquad
  s_y = \frac{\mu\,u_y}{\lambda\,u_x}\,s_x+
  \frac{(\lambda-\mu)\,u_{xy}}{\lambda\,u_x}\,s
\end{equation}
for Equation~\eqref{eq:1}.  Note that $s$ is a solution to the linearization
\eqref{linearization_eq}, \eqref{eq:9}
of VWE.  Now put
\begin{equation}
  \label{s_expansion}
  s = \sum \limits_{n=-\infty}^{\infty} s_n\, \mu^n.
\end{equation}
Since system \eqref{linearization_eq}, \eqref{eq:9} is independent of $\mu$, each $s_n$ is a solution to
this system as well. Substituting~\eqref{s_expansion} to
system~\eqref{linearization_eq}, \eqref{eq:9} yields
\begin{align*}
  s_{n+1,t} =&
               \frac{u_{tx}}{u_x}\,s_{n+1} -
               s_{n,t}+\frac{\lambda+1}{\lambda}\,\frac{u_t}{u_x}\,s_{n,x}
               -\frac{1}{\lambda}\,\frac{u_{tx}}{u_x}\,s_{n},
  \\[2mm]
  s_{n+1,y} =&
               \frac{u_{xy}}{u_x}\,s_{n+1} +\frac{1}{\lambda}\,\frac{u_y}{u_x}\,s_{n,x}
               -\frac{1}{\lambda}\,\frac{u_{xy}}{u_x}\,s_{n}.
\end{align*}
Relabeling $s_{n}=\eta$ and $s_{n+1} =\zeta$, we obtain the result.
\end{proof}

Thus, $\mathcal{R}_+$ and $\mathcal{R}_-$ are mutually inverse recursion
operators.

\subsection{Action of recursion operators}
\label{sec:acti-recurs-oper}

We now describe the action of the operators~$\mathcal{R}_+$
and~$\mathcal{R}_-$ in detail. First of all, it immediately follows
from~\eqref{eq:16} that
\begin{equation*}
  \mathcal{R}_+(0) = \xi_0^+(X),\qquad\mathcal{R}_-(0) = \xi_0^-(Y)
\end{equation*}
and thus all subsequent actions of $\mathcal{R}_{+}$ and $\mathcal{R}_{-}$ are
defined modulo addition of $\xi_0^+(X)$ and $\xi_0^-(Y)$, respectively.

Further, one has
\begin{align*}
  &\mathcal{R}_+(\phi_1(T)) = -\phi_1(T),
  &&\mathcal{R}_-(\phi_1(T)) = -\phi_1(T),\\
  &\mathcal{R}_+(\phi_4(U)) = \lambda^{-1}\phi_4(U),
  &&\mathcal{R}_-(\phi_4(U)) = \lambda\phi_4(U)
\end{align*}
and
\begin{equation*}
  \begin{array}{lll}
    \mathcal{R}_+(\xi_i^+(X))=\xi_{i+1}^+(X),\ i\geq0,
    &\mathcal{R}_+(\xi_i^-(X))=\xi_{i-1}^-(X),\ i\geq1,
    &\mathcal{R}_+(\xi_0^-(X))=0;\\[2mm]
    \mathcal{R}_-(\xi_i^-(X))=\xi_{i+1}^-(X),\ i\geq0,
    &\mathcal{R}_-(\xi_i^+(X))=\xi_{i-1}^+(X),\ i\geq1,
    &\mathcal{R}_-(\xi_0^+(X))=0.
  \end{array}
\end{equation*}
Finally,
\begin{equation}\label{eq:17}
  \begin{array}{lll}
    \mathcal{R}_+(\psi_i^+)=\psi_{i+1}^+,\ i\geq 2,
    &\mathcal{R}_+(\psi_i^-)=\psi_{i-1}^-,\ i\geq3,
    &\mathcal{R}_+(\psi_2^-)=\psi_1;\\[2mm]
    \mathcal{R}_-(\psi_i^-)=\psi_{i+1}^-,\ i\geq2,
    &\mathcal{R}_-(\psi_i^+)=\psi_{i-1}^+,\ i\geq3,
    &\mathcal{R}_-(\psi_2^+)=\psi_1.
  \end{array}
\end{equation}
The new shadow that arises in Equations~\eqref{eq:17} ``lives'' in the
Whitney product of~$\tau^+$ and~$\tau^-$ and has the form
\begin{equation}
  \label{eq:18}
  \psi_1 = \lambda\,u_x\,q_1+u_y\,r_1
\end{equation}
and will be studied in Subsection~\ref{sec:master-symmetry}.

The following diagram illustrates the above described actions:
\begin{equation*}\xymatrixcolsep{1.6pc}\xymatrixrowsep{1.3pc}
  \xymatrix{
    &&\phi_1(T)\,\ar@{->}@<1ex>[r]^-{\mathcal{R}_{+}}
    &\ar@{->}[l]^-{\mathcal{R}_{-}}\,\phi_1(T),&&
        \phi_4(U)\,\ar@{->}@<1ex>[r]^-{\mathcal{R}_{+}}
        &\ar@{->}[l]^-{\mathcal{R}_{-}}\,\phi_4(U) ,
    \\
    \dots\ \ar@{->}@<1ex>[r]^-{\mathcal{R}_{+}}
    &\ar@{->}[l]^-{\mathcal{R}_{-}}\,
    \xi^{-}_{2}(Y)\,\ar@{->}@<1ex>[r]^-{\mathcal{R}_{+}}
    &\ar@{->}[l]^-{\mathcal{R}_{-}}\,\xi_1^{-}(Y)\,\ar@<1ex>[r]^-{\mathcal{R}_{+}}
    &\ar[l]^-{\mathcal{R}_{-}}\,\xi_0^-(Y)\,\ar@<1ex>[r]^-{\mathcal{R}_{+}}
    &\ar[l]^-{\mathcal{R}_{-}}\,0\,\ar@<1ex>[r]^-{\mathcal{R}_{+}}
    &\ar[l]^-{\mathcal{R}_{-}}\,\xi_0^+(X)\,\ar@<1ex>[r]^-{\mathcal{R}_{+}}
    &\ar[l]^-{\mathcal{R}_{-}}\,\xi_1^{+}(X)\,\ar@<1ex>[r]^-{\mathcal{R}_{+}}
    &\ar[l]^-{\mathcal{R}_{-}}\,\xi_2^{+}(X)\,
    \ar@{->}@<1ex>[r]^-{\mathcal{R}_{+}}
    &\ar@{->}[l]^-{\mathcal{R}_{-}}\ \dots
    \\
    \dots\ \ar@{->}@<1ex>[r]^-{\mathcal{R}_{+}}
    &\ar@{->}[l]^-{\mathcal{R}_{-}}\,\psi^{-}_{4}\,\ar@{->}@<1ex>[r]^-{\mathcal{R}_{+}}
    &\ar@{->}[l]^-{\mathcal{R}_{-}}\,\psi_3^{-}\,\ar@<1ex>[r]^-{\mathcal{R}_{+}}
    &\ar[l]^-{\mathcal{R}_{-}}\,\psi_2^{-}\,\ar@<1ex>[r]^-{\mathcal{R}_{+}}
    &\ar[l]^-{\mathcal{R}_{-}}\,\psi_1\,\ar@<1ex>[r]^-{\mathcal{R}_{+}}
    &\ar[l]^-{\mathcal{R}_{-}}\,\psi_2^{+}\,\ar@<1ex>[r]^-{\mathcal{R}_{+}}
    &\ar[l]^-{\mathcal{R}_{-}}\,\psi_3^{+}\,\ar@<1ex>[r]^-{\mathcal{R}_{+}}
    &\ar[l]^-{\mathcal{R}_{-}}\,\psi_4^{+}\,
    \ar@{->}@<1ex>[r]^-{\mathcal{R}_{+}}
    &\ar@{->}[l]^-{\mathcal{R}_{-}}\ \dots
  }
\end{equation*}

\subsection{Master-symmetry}
\label{sec:master-symmetry}

Let us describe the lift
\begin{equation*}
  \Psi_1 = \psi_1\frac{\partial}{\partial u} + \sum_{i}\left(\psi_1^{+,i}
    \frac{\partial}{\partial q_i} + \psi_1^{-,i}\frac{\partial}{\partial
      r_i}\right)
\end{equation*}
of the shadow~$\psi_1$ to the Whitney product~$\tau^+\oplus\tau^-$. To
this end we set
\begin{align*}
  \psi_1^{+,i}&=  iq_i + (i - 1)q_{i-1} + r_1 q_{i,y} - \lambda\big((i + 1)q_{i+1} +
            iq_i - q_1 q_{i,x} \big),\\
  \psi_1^{-,i}&=-(i + 1)r_{i+1} - ir_i + r_1 r_{i,y} + \lambda\big(ir_i + (i -
            1)r_{i-1} + q_1 r_{i,x} \big),
\end{align*}
for all $i\geq1$, and this is the desired lift. Then the commutators
of~$\Psi_1$ with the already constructed symmetries are as follows:
\begin{equation*}
  \{\Psi_1,\Psi_i^+\} =
  \begin{cases}
    \lambda\Psi_3^+ - (2\lambda -1)\Psi_2^+ - \dfrac{1}{\lambda}\Psi_1,&
    i=2,\\
    \lambda(i-1)\Psi_{i+1}^+ - ((\lambda-1)i+1)\Psi_i^+ - i\Psi_{i-1}^+ -
    \dfrac{1}{\lambda} \Psi_1,&i>2,
  \end{cases}
\end{equation*}
and
\begin{equation*}
  \{\Psi_1,\Psi_i^-\} =
  \begin{cases}
    \Psi_3^- +(\lambda-2)\Psi_2^- - \lambda\Psi_1,&i=2,\\
    (i-1)\Psi_{i+1}^- + ((\lambda-1)i - \lambda)\Psi_i^- -
    \lambda(i\Psi_{i-1}^- + \Psi_1),&i>2.
  \end{cases}
\end{equation*}
Further, we have
\begin{equation*}
  \{\Psi_1,\Xi_i^+(X)\} =
  \begin{cases}
    i\lambda\Xi_{i+1}^+(X) - \big(i(\lambda-1)+1\big)\Xi_i^+(X) -
    (i-1)\Xi_{i-1}^+(X),&i>0,\\
    0,&i=0,\\
    i\Big(\lambda \Xi_{i+1}^+(X) - (\lambda - 1)\Xi_i^+ -
    \Xi_{i-1}^+(X)\Big),&i \leq  -1,
  \end{cases}
\end{equation*}
and
\begin{equation*}
  \{\Psi_1,\Xi_i^-(Y)\} =
  \begin{cases}
    i\Xi_{i+1}^-(Y)  +\big(i(\lambda-1) - \lambda\big)\Xi_{i}^-(Y) - \lambda(i
    - 1)\Xi_{i-1}^-(Y)),&i>0,\\
    0,&i=0,\\
    (i - 1)\Big(\Xi_{i-1}^-(Y)
    + (\lambda - 1)\Xi_i^-(Y ) - \lambda \Xi_{i+1}^-(Y)\Big),&
    i \leq -1.
  \end{cases}
\end{equation*}
Note finally, that~$\{\Psi_1,\Phi_1(T)\} = \{\Psi_1,\Phi_4(U)\} = 0$.

Thus we see that~$\Psi_1$ plays the role of a master-symmetry:
taking~$\Psi_2^+$, $\Psi_2^-$, $\Xi_{\pm 1}^+(X)$, $\Xi_{\pm 1}^-(Y)$ for
``seeds'' and acting by~$\{\Psi_1,\cdot\}$, we can obtain the entire
hierarchies~$\Psi_i^+$, $\Psi_i^-$, $i>2$, $\Xi_i^+(X)$, $\Xi_i^-(Y)$, $i>1$,
and $\Xi_i^+(X)$, $\Xi_i^-(Y)$, $i<-1$.

To conclude, let us compare briefly the Lie algebra structures of nonlocal
symmetries for all the five linearly degenerate 3D equations studied
in~\cite{B-K-M-V-2018} and here. All these algebras are
infinite-dimensional. For the the rdDym equation
$u_{ty} = u_x u_{xy} - u_y u_{xx}$, the 3D Pavlov equation
\eqref{Pavlov_eq}, and the universal hierarchy
equation $u_{yy} = u_t u_{xy} - u_y u_{tx}$ they are graded. The symmetry
algebra of the modified Veronese web equation
$u_{ty} = u_t u_{xy} - u_y u_{tx}$ is filtered (almost-graded). The
corresponding algebra for the VWE~\eqref{eq:1} seemingly admits no reasonable
grading or filtering and contains a real irremovable parameter~$\lambda$. It
will be interesting to study the properties of this algebra in more detail
elsewhere.

\section*{Acknowledgments}

Computations were supported by the \textsc{Jets} software,~\cite{Jets}.

\end{document}